\documentclass[12pt]{article}
\usepackage{fullpage}
\usepackage{latexsym}
\usepackage{amsmath,amsfonts,amssymb,amsthm,enumerate,dsfont,color}

\usepackage{tikz}
\usetikzlibrary{calc}
\usetikzlibrary{arrows}
\usetikzlibrary{shadows}
\usetikzlibrary{positioning}
\usetikzlibrary{decorations.pathmorphing,shapes}
\usetikzlibrary{shapes,backgrounds}
\usepackage{mathrsfs}
\usepackage{float}

\xdefinecolor{olive}{cmyk}{0.64,0,0.95,0.4}
\xdefinecolor{blueberry}{cmyk}{0.44,0.23,0.0,0.18}
\xdefinecolor{maroon}{cmyk}{0.00,1.0,1.0,0.498}

\newcommand{\ignore}[1]{}

\newtheoremstyle{postnum}
  {\topsep}
  {\topsep}
  {\slshape}
  {0pt}
  {\bfseries}
  {:}
  { }
  {\thmname{#1}\thmnote{ (#3)}}

\theoremstyle{postnum}

\newtheoremstyle{prenum}
  {\topsep}
  {\topsep}
  {\slshape}
  {0pt}
  {\bfseries}
  {.}
  { }
  {\thmnumber{#2}\thmname{ #1}\thmnote{ (#3)}}

\theoremstyle{prenum}

\newtheorem{theorem}{Theorem}[section]
\newtheorem{lemma}[theorem]{Lemma}

\newtheorem{proposition}[theorem]{Proposition}
\newtheorem{corollary}[theorem]{Corollary}
\newtheorem{example}[theorem]{Example}

\newcommand{\etal}{{\em et al.~}}

\newcommand{\ZZ}{\mathbb{Z}}
\newcommand{\RR}{\mathbb{R}}
\newcommand{\QQ}{\mathbb{Q}}
\newcommand{\CC}{\mathbb{C}}

\newcommand{\ii}{\mathtt{i}}

\DeclareMathOperator{\spn}{span}
\DeclareMathOperator{\spec}{Sp}

\newcommand{\eket}[1]{\mathbf{e}_{#1}}
\newcommand{\ebra}[1]{\mathbf{e}_{#1}^{T}}

\newcommand{\ket}[1]{| #1 \rangle}

\newcommand{\HH}{\mathcal{H}}

\newcommand{\zo}{\{0,1\}}

\newcommand{\Q}{\mathds{Q}}

\newcommand{\ee}{\mathbf{e}}

\newcommand{\vv}{\mathbf{v}}

\newcommand{\0}{\textbf{0}}

\newcommand{\zero}{\texttt{O}}

\newcommand{\pmat}[1]{\begin{pmatrix}#1\end{pmatrix}}

\newcommand{\ov}{\overline}

\newcommand\gustone[4][]{%
    \begin{scope}
        \fill[stone,#2-stone] (#3,#4) circle (0.45);
        \clip (#3,#4) circle (0.45);
        \shade[#2-highlight] (-0.15+#3,0.5+#4) circle (0.7);
    \end{scope}
    \node[#2-number] at (#3,#4) {};
}

\newif\ifnotesw\noteswtrue
   {\ifnotesw\marginpar[\hfill\(\top\)]{\(\top\)}\fi}%
   {\ifnotesw\marginpar[\hfill\(\bot\)]{\(\bot\)}\fi}

\newcommand{\mnote}[1]%
    {\ifnotesw\marginpar%
        [{\scriptsize\begin{minipage}[t]{\marginparwidth}
        \raggedleft#1%
                        \end{minipage}}]%
        {\scriptsize\begin{minipage}[t]{\marginparwidth}
        \raggedright#1%
                        \end{minipage}}%
    \fi}

\title{Quantum Fractional Revival on Graphs}
\author{
Ada Chan \\ \scriptsize{Department of Mathematics and Statistics,} \\ \scriptsize{York University.}
\and
Gabriel Coutinho\thanks{Corresponding author: \texttt{gabriel@dcc.ufmg.br}}\\ \scriptsize{Departamento de Ci\^{e}ncia da Computa\c{c}\~{a}o,} \\ \scriptsize{Universidade Federal de Minas Gerais.}
\and
Christino Tamon \\ \scriptsize{Department of Computer Science,}\\ \scriptsize{Clarkson University.}
\and
Luc Vinet \\ \scriptsize{Centre de Recherches Math\'{e}matiques,} \\ \scriptsize{Universit\'e de Montr\'eal.}
\and
Hanmeng Zhan \\ \scriptsize{Department of Combinatorics and Optimization,} \\ \scriptsize{University of Waterloo.}
}
\date{\today}
\begin{document}
\maketitle
\begin{abstract}
Fractional revival is a quantum transport phenomenon important for entanglement generation in spin networks. 
This takes place whenever a continuous-time quantum walk maps 
the characteristic vector of a vertex to a superposition of the characteristic vectors of 
a subset of vertices containing the initial vertex. 
A main focus will be on the case when the subset has two vertices.
We explore necessary and sufficient spectral conditions for graphs to exhibit fractional revival.
This provides a characterization of fractional revival in paths and cycles.
Our work builds upon the algebraic machinery developed for related quantum transport phenomena 
such as state transfer and mixing, and it reveals a fundamental connection between them.
\end{abstract}

\newpage


\section{Introduction}

Continuous-time quantum walk is a fundamental technique in quantum information and computation.
It was used by Farhi and Gutmann \cite{fg98} to study quantum algorithmic problems on graphs.
This has led to the design of a continuous-time quantum walk algorithm which exhibits an exponential
speedup over any classical algorithm (see Childs \etal \cite{ccdfgs03}).
The connection between quantum walk in the continuous-time and discrete-time models 
was described by Childs \cite{c10}.

In quantum information, continuous-time quantum walk has proved important for studying transport 
problems in quantum spin networks. 
This was initiated by Bose \cite{b03} who studied state transfer in quantum spin networks. 
Further work along these lines were described by Christandl \etal \cite{cdel04,acde04,cddekl05}
and a comprehensive survey was given by Kay \cite{k11}.

A quantum spin network on a graph with $n$ vertices is obtained by attaching a qubit to 
each vertex of the graph. Here, a quantum state is a normalized vector in $(\CC^{2})^{\otimes n}$. 
Specifically, if $\ket{0},\ket{1}$ denote the basis states of our qubit, 
then a quantum state of the network is a normalized vector in
$\spn\{\ket{x} : x \in \zo^{n}\}$, where $\ket{x} = \bigotimes_{i=1}^{n} \ket{x_{i}}$.
Suppose $X=(V,E)$ is a weighted graph on $n$ vertices represented by a Hermitian matrix $A = (a_{u,v})$.
Consider a Hamiltonian $H$ of the form\footnote{Here,
$\sigma_{X}^{(u)} = \bigotimes_{v} M_{v}$, where $M_{v} = I$ if $v \neq u$, and $M_{u} = \sigma_{X}$. 
A similar convention applies to $\sigma_{Y},\sigma_{Z}$.}
\begin{equation}
H = \sum_{(u,v) \in E} a_{u,v} (\sigma_{X}^{(u)}\sigma_{X}^{(v)} + \sigma_{Y}^{(u)}\sigma_{Y}^{(v)})
		+ \sum_{u \in V} a_{u,u}\sigma_{Z}^{(u)}.
\end{equation}
Since $H$ commutes with $H_{0} = \sum_{u \in V} \sigma_{Z}^{(u)}$, 
the Schr\"{o}dinger evolution
\begin{equation} \label{eqn:schrodinger}
\ket{\Psi(t)} = e^{-\ii tH}\ket{\Psi(0)}
\end{equation}
preserves the number of qubits in the excited state.

In a quantum transport involving two vertices $a,b \in V$, 
we consider a state with an arbitrary qubit $\ket{\psi}$ at $a$ 
and qubit $\ket{0}$ at the other vertices.
If we label the vertices with integers in such a way that
$a$ is the leftmost one and $b$ is the rightmost one, then
the initial state is given by
\begin{equation}
\ket{\Psi(0)} = \ket{\psi}_{a} \ket{0}^{\otimes (n-2)} \ket{0}_{b}. 
\end{equation} 
For {\em perfect state transfer} from $a$ to $b$, 
we ask if there a time $\tau$ for which the qubit $\ket{\psi}$ has been transported to $b$; that is,
up to a global phase, we have
\begin{equation} \label{eqn:pst_question}
\ket{\Psi(\tau)} 
	= \ket{0}_{a} \ket{0}^{\otimes (n-2)} \ket{\psi}_{b}
\end{equation}
If $\ket{\psi} = \mu\ket{0} + \nu\ket{1}$, this is equivalent to requiring
\begin{equation}
e^{-\ii \tau H} \left(\mu\ket{0}^{\otimes n} 
		+ \nu\ket{1}_{a} \ket{0}^{\otimes (n-2)} \ket{0}_{b}\right) 
	= \mu\ket{0}^{\otimes n} 
		+ \nu \ket{0}_{a} \ket{0}^{\otimes (n-2)} \ket{1}_{b} 
\end{equation}
or, if $\ket{a},\ket{b}$ denote 
$\ket{1}_{a}\ket{0}^{\otimes (n-2)}\ket{0}_{b}$,
$\ket{0}_{a}\ket{0}^{\otimes (n-2)}\ket{1}_{b}$, respectively, 
then
\begin{equation} \label{eqn:one-excitation}
e^{-\ii\tau H}\ket{a} = \gamma\ket{b}.
\end{equation}
In quantum information, it is known that entanglement is a useful and important resource.
A relevant quantum transport phenomenon for entanglement generation is {fractional revival}.
We say {\em fractional revival} occurs at $a$ and $b$ at time $\tau$ if
\begin{equation} \label{eqn:fr_in_spin}
e^{-\ii\tau H}\ket{a} = \alpha\ket{a} + \beta\ket{b}
\end{equation}
for complex scalars $\alpha$ and $\beta \neq 0$ with $|\alpha|^{2} + |\beta|^{2} = 1$.
Tracing out the other qubits, the resulting quantum state in \eqref{eqn:fr_in_spin} is equivalent 
to the entangled state 
\begin{equation}
\alpha\ket{1}_{a}\ket{0}_{b} + \beta\ket{0}_{a}\ket{1}_{b}.
\end{equation}
Note that, in the spin network picture, we are not requiring $e^{-\ii tH}$ maps
$\ket{\psi}_{a}\ket{0}_{b}$ to $\ket{\psi}_{a}\ket{\psi}_{b}$ as this violates the
No Cloning theorem.

Note that the evolution in \eqref{eqn:one-excitation} stays in the single-excitation subspace 
$\spn\{\eket{u} : u \in V\}$, where $\eket{u}$ is the characteristic vector of vertex $u$.
So, if $\ket{\Psi(0)}$ belongs to this subspace, the unitary evolution is simply
\begin{equation} \label{eqn:unitary-walk}
U(t) = e^{-\ii tA}
\end{equation}
which is called a continuous-time quantum walk on $X$ (see \cite{fg98}).
We adopt this framework throughout our work.

Fractional revival had been studied in quantum spin networks 
(see 
Chen \etal \cite{css07},
Banchi \etal \cite{bcb15},
Genest \etal \cite{gvz16}, 
and
Christandl \etal \cite{cvz17}). 
The graphs they studied may be viewed as weighted paths 
and their analysis reveal the power of the weighting schemes via orthogonal polynomials 
(for example, Krawtchouk polynomials).
In this work, we focus on quantum fractional revival in graphs (whose adjacency matrices are $\zo$-valued).
The goal is to understand the role of the underlying graph structure on fractional revival 
while focusing less on the effect of the weighting scheme. 
Most graphs we consider are unweighted although in some cases we draw comparisons 
with their weighted cousins.
Our work reveals a fundamental connection and useful interplay between fractional revival 
and other quantum transport phenomena.

\begin{figure}[t]
\begin{center}
\begin{tikzpicture}[
    scale=0.75,
    stone/.style={}, 
    black-stone/.style={black!80},
    black-highlight/.style={outer color=black!80, inner color=black!30},
    black-number/.style={white},
    white-stone/.style={white!70!black},
    white-highlight/.style={outer color=white!70!black, inner color=white},
    white-number/.style={black},
    decoration=snake]

\gustone[1]{black}{-3}{+1.5}
\gustone[2]{black}{0}{+1.5}
\gustone[3]{black}{+3}{+1.5}

\draw[line width = 2pt, red, ->] (-3,2) -- ++(0,0.75);
\draw[line width = 2pt, blue, <-] (0,2) -- ++(0,0.75);
\draw[line width = 2pt, blue, <-] (3,2) -- ++(0,0.75);

\path (-2.75, 1.5) node (x1) [scale=0.9] {};
\path (-0.25, 1.5) node (x2) [scale=0.9] {};
\path (+0.25, 1.5) node (y1) [scale=0.9] {};
\path (+2.75, 1.5) node (y2) [scale=0.9] {};
\draw[decorate, very thick]
    (x1) -- (x2);
\draw[decorate, very thick]
    (y1) -- (y2);

\end{tikzpicture}
\quad \quad \quad
\begin{tikzpicture}[
    scale=0.75,
    stone/.style={}, 
    black-stone/.style={black!80},
    black-highlight/.style={outer color=black!80, inner color=black!30},
    black-number/.style={white},
    white-stone/.style={white!70!black},
    white-highlight/.style={outer color=white!70!black, inner color=white},
    white-number/.style={black},
    decoration=snake]

\gustone[1]{black}{-3}{-1.5}
\gustone[2]{black}{0}{-1.5}
\gustone[3]{black}{+3}{-1.5}

\draw[line width = 2pt, blue, <-] (-2.8,-1) -- ++(0,0.75); 
\draw[line width = 2pt, red, ->] (-3.2,-1) -- ++(0,0.75); 
\draw[line width = 2pt, blue, <-] (-0.2,-1) -- ++(0,0.75); 
\draw[line width = 2pt, blue, <-] (+0.2,-1) -- ++(0,0.75); 
\draw[line width = 2pt, blue, <-] (2.8,-1) -- ++(0,0.75);
\draw[line width = 2pt, red, ->] (3.2,-1) -- ++(0,0.75);

\path (-2.75, -1.5) node (x1) [scale=0.9] {};
\path (-0.25, -1.5) node (x2) [scale=0.9] {};
\path (+0.25, -1.5) node (y1) [scale=0.9] {};
\path (+2.75, -1.5) node (y2) [scale=0.9] {};
\draw[decorate, very thick]
    (x1) -- (x2);
\draw[decorate, very thick]
    (y1) -- (y2);

\end{tikzpicture}
\caption{Fractional revival on the path $P_{3}$ of three vertices (sites).
Up to normalization and tracing out the middle qubit,
the initial state (left) is $\ket{1}\ket{0}$ 
and
the final state (right) is $\ket{1}\ket{0} + \ket{0}\ket{1}$.
}
\end{center}
\end{figure}
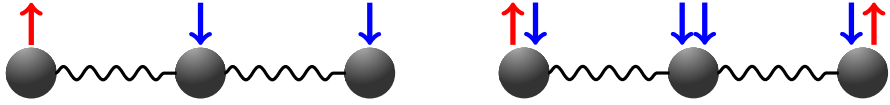

An outline of this paper is as follows.
In Section \ref{section:preliminaries}, we briefly provide a relevant background of graph theory 
and formally define the basic notions of state transfer in graphs.
In Section \ref{section:construction}, we describe generic constructions of graphs with
fractional revival from graphs with other quantum transport properties.
One of our constructions answers an open question posed by Genest \etal \cite{gvz16}
while another one provides a combinatorial variant of a spectral based method
described in the same work.

In Section \ref{section:properties}, we observe basic properties of fractional revival in graphs.
First, we show weak symmetry: if fractional revival exists from $a$ to $b$, 
then it also exists from $b$ to $a$ although not necessarily preserving the original superposition.
Second, we show that if two vertices are involved in a fractional revival, then 
the projections of their characteristic vectors onto any eigenspace of the graph are 
parallel as vectors (here, the vertices are called {\em parallel} \cite{g12}).
In Section \ref{section:cospectral}, we exploit this spectral property for characterizing 
fractional revival between cospectral vectices in a graph. 

Two vertices $a$ and $b$ in a graph $X$ are called {\em cospectral} if
the vertex deleted subgraphs $X \setminus a$ and $X \setminus b$ share the same spectra.
We show that if fractional revival occurs between cospectral vertices,
then perfect or pretty good (almost) state transfer occurs between them 
or both vertices are periodic (perfect state transfer to themselves).
This shows that, in the presence of cospectrality, fractional revival is a sufficient condition 
for these relevant quantum transport phenomena.
Moreover, we observe that, when fractional revival occurs between cospectral vertices,
the relevant eigenvalues are either all integers or all algebraic integers from a quadratic 
field. In the latter case, the specific form for the eigenvalues which belong to a quadratic field
has an extra degree of freedom when compared to the perfect state transfer case
(as proved by Godsil \cite{g12b}).

In Section \ref{section:equitable}, we observe that fractional revival in a graph is preserved 
under taking a graph quotient modulo an equitable partition. Moreover, any graph automorphism 
which fixes a vertex involved in a fractional revival must also fix the other vertex.
Building upon these observations, in the last two sections, we characterize fractional revival 
in two infinite families of graphs.
In Section \ref{section:bipartite}, we show that fractional revival occurs on a cycle
if and only if the cycle has order four or six (see Figure \ref{fig:cycles}).
In Section \ref{section:path}, we prove that fractional revival occurs on a path
if and only if the path has order two, three or four (see Figure \ref{fig:paths}).
It is curious to note that these nearly match the known orders of cycles and paths 
with perfect state transfer (see \cite{cddekl05,g12}).


\section{Preliminaries} \label{section:preliminaries}

\par\noindent{\bf Notation}. 
The identity and all-one matrices of order $n$ are denoted $I_{n}$ and $J_{n}$, respectively;
but we will omit the subscript $n$ whenever it is clear from context.
For a vector $\mathbf{v}$, we denote its $i$th entry as $\mathbf{v}_{i}$.
We use $\ee_{i}$ to denote the unit vector that is $1$ at position $i$ and is $0$ elsewhere.

We briefly review some relevant background from graph theory (see Godsil and Royle \cite{gr01}).
For a graph $X$, we denote its set of vertices as $V(X)$ and its set of edges 
(which are pairs of vertices) as $E(X)$.
The adjacency matrix $A(X)$ of $X$ is a matrix of order $|V(X)|$ defined as 
$A(X)_{a,b} = 1$, if $(a,b) \in E(X)$, and $0$ otherwise.
The spectrum $\spec(X)$ of $X$ is the set of eigenvalues of $A(X)$.
We say $X$ is undirected if $A(X)$ is symmetric and $X$ is simple if $A(X)$ has zero diagonal.
We call $X$ bipartite if its vertex set admits a bipartition $V_{1} \cup V_{2}$ where each
edge intersects both $V_{1}$ and $V_{2}$.

For two graphs $X$ and $Y$, the Cartesian product $X \Box Y$ is a graph over $V(X) \times V(Y)$ 
whose edges consist of pairs $(x_{1},y_{1})$ and $(x_{2},y_{2})$ which
are adjacent if either $(x_{1},x_{2}) \in E(X)$ with $y_{1}=y_{2}$ 
or $(y_{1},y_{2}) \in E(Y)$ with $x_{1}=x_{2}$.
A relevant fact is that $A(X \Box Y) = A(X) \otimes I + I \otimes A(Y)$.
The union $X \cup Y$ of $X$ and $Y$ is a graph with vertex set $V(X) \cup V(Y)$
whose edge set is $E(X) \cup E(Y)$. The complement $\overline{X}$ of graph $X$ is a graph
with the same vertex set but whose edge set includes all pairs of distinct vertices which 
are not adjacent in $X$.
The join $X + Y$ of $X$ and $Y$ is the graph whose complement is $\overline{X} \cup \overline{Y}$.

We define some standard families of graphs relevant to our work.
The path $P_{n}$ has vertex set $\{1,\ldots,n\}$ where $(i,j) \in E(P_{n})$ if $|i-j| = 1$.
The cycle $C_{n}$ has vertex set $\ZZ/n\ZZ$ with $(i,j) \in E(C_{n})$ if $i-j \equiv \pm 1\pmod{n}$.
The clique (or complete graph) $K_{n}$ has edges between each pair of distint vertices,
which implies $A(K_{n}) = J_{n}-I_{n}$.
The star $K_{1,n}$ is the graph $K_{1} + \overline{K_{n}}$.
The double cone of $X$ is given by $\overline{K_{2}} + X$.
The $d$-dimensional cube $Q_{d}$ is defined recursively as $Q_{d} = K_{2} \Box Q_{d-1}$, if $d > 1$,
and $Q_{1} = K_{2}$.

Next, we formally define pertinent notions related to quantum walk and state transfer.
For a graph $X$, the {\em continuous-time quantum walk} (or transition) matrix of $X$ is given by
\begin{equation}
U_{X}(t) = \exp(-\ii t A(X))
\end{equation}
where $A(X)$ is the adjacency matrix of $X$. Whenever it is clear from context,
we drop the subscript and simply use $U(t)$.
Let $a$ and $b$ be two distinct vertices of $X$. 
We say that $X$ admits {\em fractional revival from $a$ to $b$} at time $\tau \in \RR \setminus \{0\}$
if for some $\alpha, \beta \in \CC$, with $|\alpha|^2 + |\beta|^2 = 1$ 
{\color{black} and $\beta\neq 0$}, 
we have \begin{align} \label{eqn:fr_def}
U(\tau) \ee_a = \alpha \ee_a + \beta \ee_b. 
\end{align}
In this case, we also say that {\em $(\alpha,\beta)$-revival occurs from $a$ to $b$ at time $\tau$}.
The fractional revival is called {\em balanced} if $|\alpha| = |\beta|$.
We say $X$ has $(\alpha,\beta)$-revival if it has $(\alpha,\beta)$-revival from each of its vertices 
at the same time. This holds if there is a permutation matrix $T$ 
(with no fixed points) where for some time $\tau$ we have
\begin{equation}
U(\tau) = \alpha I + \beta T.
\end{equation}

We define several other quantum transport properties.
The graph $X$ is called {\em periodic at vertex $a$} at time $\tau$ 
if $\beta = 0$ in \eqref{eqn:fr_def}.   
{\color{black} In this case, any vertex $b$ satisfies \eqref{eqn:fr_def} and we have chosen not to
consider this as a special case of fractional revival.}
We say $X$ has {\em perfect state transfer} from $a$ to $b$ at time $\tau$
if $\alpha = 0$ in \eqref{eqn:fr_def}.
Finally, we say $X$ has {\em (instantaneous) uniform mixing} at time $\tau$ if
$U(\tau)$ is a flat matrix whose entries all have the same magnitude.
(See Godsil \cite{g12} for a survey of these notions.)


\section{Generic Constructions} \label{section:construction}

In this section, we describe methods for constructing graphs with fractional revival from 
graphs with perfect state transfer, vertex periodicity, and/or uniform mixing. 

Our first theorem shows how to construct generalized fractional revival from graphs with 
periodicity and instantaneous uniform mixing.
To this end, we consider a generalized notion of fractional revival.
Given a vertex $a$ and a subset of vertices $B \subset V(X)$ where $a \in B$, 
we say $X$ has {\em generalized fractional revival} from $a$ to $B$ at time $\tau$ if we have
$\ebra{b}U(\tau)\eket{a} \neq 0$ if and only if $b \in B$.

\begin{theorem} \label{thm:ium_to_fr}
Let $X$ be a graph that is periodic at vertex $a$ at time $\tau$.
Let $Y$ be a graph with instantaneous uniform mixing at time $\tau$.
Then, for any vertex $u$ of $Y$, the graph $X \Box Y$ has generalized fractional revival 
from $(a,u)$ to the vertices $\{(a,v) : v \in V(Y)\}$ at time $\tau$.

\begin{proof}
Since $U_{X \Box Y}(t) = U_{X}(t) \otimes U_{Y}(t)$, we have
\begin{equation}
U_{X \Box Y}(\tau)
	= \gamma \pmat{1 & \zero^T \\ \zero & B} \otimes W,
\end{equation}
where $\gamma$ is a complex unimodular number, $B$ is a symmetric matrix and $W$ is a unitary flat matrix.
From this we see that 
\begin{equation}
U_{X \Box Y}(\tau) \ee_{(a,u)} = \frac{\gamma}{\sqrt{|V(Y)|}} \sum_{v \in V(Y)} \alpha_{v}\ee_{(a,v)}
\end{equation}
where $\alpha_{v}$ is complex unimodular for each $v \in V(Y)$.
Hence, there is fractional revival from vertex $(a,u)$ to the set of vertices $\{(a,v) : v \in V(Y)\}$ 
at time $\tau$.
\end{proof}
\end{theorem}

\begin{corollary} \label{cor:bunkbed}
Suppose $X$ is periodic at vertex $a$ with period $\tau < \pi/2$. 
Then the Cartesian product $X\square K_2$ has fractional revival from vertex $(a,u)$ 
to vertices $(a,u)$ and $(a,v)$ at time $\tau$, where $u$ and $v$ are the vertices of $K_2$.
\end{corollary}

In Theorem \ref{thm:ium_to_fr}, if $|V(Y)| > 2$, then the generalized fractional revival occurs 
among a set of more than two vertices. This answers a question posed by Genest \etal \cite{gvz16}.
We provide an explicit family of examples in the following.

\begin{example} (Fractional revival from periodicity and uniform mixing) \\
The star $K_{1,n}$ is periodic at time $t = \pi/\sqrt{n}$ and
the $d$-cube $Q_{d}$ has instantaneous uniform mixing at time $t = \pi/4$ 
(see Moore and Russell \cite{mr02}).
Then, the graph $K_{1,16} \Box Q_{d}$ has generalized fractional revival among 
$2^{d}$ vertices at time $\pi/4$, for all $d \ge 1$.
When the cube is $K_{2}$ or $d=1$, we have balanced fractional revival at time $\pi/4$.
\end{example}

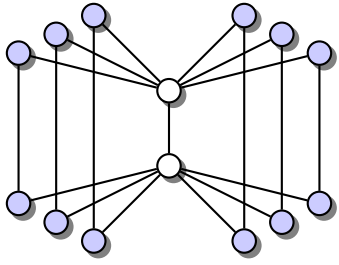
\begin{figure}[t]
\begin{center}
\begin{tikzpicture}[
    scale=0.5,
    stone/.style={},
    black-stone/.style={black!80},
    black-highlight/.style={outer color=black!80, inner color=black!30},
    black-number/.style={white},
    white-stone/.style={white!70!black},
    white-highlight/.style={outer color=white!70!black, inner color=white},
    white-number/.style={black}]


\tikzstyle{every node}=[draw, thick, shape=circle, circular drop shadow, fill={white}];
\path (0,1) node (a) [scale=0.8] {};
\path (0,-1) node (b) [scale=0.8] {};
\tikzstyle{every node}=[draw, thick, shape=circle, circular drop shadow, fill={blue!20}];
\path (-4,+2) node (a1) [scale=0.8] {};
\path (-3,+2.5) node (a2) [scale=0.8] {};
\path (-2,+3) node (a3) [scale=0.8] {};
\path (+2,+3) node (a4) [scale=0.8] {};
\path (+3,+2.5) node (a5) [scale=0.8] {};
\path (+4,+2) node (a6) [scale=0.8] {};
\draw[thick]
    (a) -- (a1)
    (a) -- (a2)
    (a) -- (a3)
    (a) -- (a4)
    (a) -- (a5)
    (a) -- (a6);
\path (-4,-2) node (b1) [scale=0.8] {};
\path (-3,-2.5) node (b2) [scale=0.8] {};
\path (-2,-3) node (b3) [scale=0.8] {};
\path (+2,-3) node (b4) [scale=0.8] {};
\path (+3,-2.5) node (b5) [scale=0.8] {};
\path (+4,-2) node (b6) [scale=0.8] {};
\draw[thick]
    (b) -- (b1)
    (b) -- (b2)
    (b) -- (b3)
    (b) -- (b4)
    (b) -- (b5)
    (b) -- (b6);
\draw[thick]
    (a) -- (b)
    (a1) -- (b1)
    (a2) -- (b2)
    (a3) -- (b3)
    (a4) -- (b4)
    (a5) -- (b5)
    (a6) -- (b6);
\end{tikzpicture}
\caption{The graph $K_{1,6} \Box K_{2}$ has fractional revival between vertices marked white.}
\label{fig:claw_bunkbed}
\end{center}
\end{figure}

\ignore{
\begin{example} (Fractional revival from periodicity) \\
Recall that the star $K_{1,n}$ is periodic at time $t = \pi/\sqrt{n}$.
Then, the irregular graph $K_{1,n} \Box K_{2}$ has fractional revival at time $\pi/\sqrt{n}$.
The fractional revival is balanced for $n=16$.
\end{example}
}

Our second theorem shows how to construct fractional revival from graphs with 
perfect state transfer by taking a union with an auxiliary graph induced by the
perfect state transfer.

\begin{theorem}
Let $X$ be a graph with perfect state transfer between vertices $a$ and $b$ at time $\tau$,
where $\tau < \pi/2$. 
Let $Y$ be a graph on the same vertex set as $X$ where $(a,b)$ is an isolated edge.
If the adjacency matrices of $X$ and $Y$ commute, 
then $X \cup Y$ has fractional revival at time $\tau$.

\begin{proof}
The adjacency matrix of $X \cup Y$ is given by $A(X \cup Y) = A(X) + A(Y)$.
Since $A(X)$ and $A(Y)$ commute, we have
\begin{equation}
U_{X \cup Y}(\tau) = U_{Y}(\tau)U_{X}(\tau).
\end{equation}
Therefore, for some complex unimodular $\gamma$, we have
\begin{eqnarray}
U_{X \cup Y}(\tau)\ee_{a} 
	& = & U_{Y}(\tau)U_{X}(\tau)\ee_{a} \\
	& = & \gamma U_{Y}(\tau)\ee_{b} \\
	& = & \gamma (\cos(\tau)\ee_{b} - \ii\sin(\tau)\ee_{a}).
\end{eqnarray}
Since $\tau < \pi/2$, this shows fractional revival occurs between $a$ and $b$.
\end{proof}
\end{theorem}

\begin{example} \label{eg:pst_to_fr}
(Fractional revival from perfect state transfer) \\
Let $Y$ be a graph with perfect state tranfer at time $\pi/2^{k}$ for $k \ge 2$.
For example, we may choose one of the graphs from the Hamming scheme $\HH(n,2)$
constructed by Chan \cite{chan}. 
Let $X = A_{n}$ be the Hamming graph that is a perfect matching containing 
the edges $\{(a,\ov{a}) : a \in \zo^{n}\}$, where 
$\ov{a}$ denote $a$ with all bits flipped.
Then, $X \cup Y$ has fractional revival at $\pi/2^{k}$ between the antipodal pair of vertices.
Note fractional revival is balanced if $k=2$.
\end{example}

The next theorem shows another method to construct a graph with fractional revival from
a graph with perfect state transfer. It is related to the {\em isospectral} transformation 
on spin networks described 
by Genest \etal \cite{gvz16}, Dai \etal \cite{dfk09}, and Kay \cite{k11}, 
but in the context of graphs.

\begin{theorem}
Suppose $Y$ has perfect state transfer between vertices $a$ and $b$ at time $\pi/2$.
Assume there is an automorphism $T$ of $Y$ with order two which swaps $a$ and $b$.
Consider the graph $X_{\theta}$ whose adjacency matrix is
\begin{equation}
A(X_{\theta}) = 
	I \otimes Y 
	+ \cos(2\theta) (\sigma_{X} \otimes I)
	+ \sin(2\theta) (\sigma_{Z} \otimes T). 
\end{equation}
(Note: $X_{0} = K_{2} \Box Y$)
Then, $X_{\theta}$ has $e^{-\ii\pi/2}(\sin(2\theta),\cos(2\theta))$-revival between $(0,a)$ and $(1,b)$ 
at time $\pi/2$.

\begin{proof}
Note that
\begin{equation}
e^{-\ii tA(X_{\theta})}
	= \exp(-\ii t(\cos(2\theta) (\sigma_{X} \otimes I) + \sin(2\theta) (\sigma_{Z} \otimes T)))
		(I \otimes e^{-\ii tA(Y)}).
\end{equation}
Since $\sigma_{X} \otimes I$ and $\sigma_{Z} \otimes T$ are anti-commuting, we have
\begin{eqnarray}
\lefteqn{\exp(-\ii t (\cos(2\theta) (\sigma_{X} \otimes I) + \sin(2\theta) (\sigma_{Z} \otimes T)))} \\
	& = &
	\cos(t) I - \ii\sin(t)(\cos(2\theta) (\sigma_{X} \otimes I) + \sin(2\theta)(\sigma_{Z} \otimes T)).
\end{eqnarray}
Therefore,
\begin{equation}
e^{-\ii tA(X_{\theta})}\ket{0,a} = 
	-\ii\sin(t)\sin(2\theta)\ket{0,a}
	+
	\cos(t)\ket{0,b}
	- 
	\ii\sin(t)\cos(2\theta)\ket{1,b}.
\end{equation}
When $t=\pi/2$, this yields
\begin{equation}
e^{-\ii (\pi/2)A(X_{\theta})}\ket{0,a} = 
	-\ii\sin(2\theta)\ket{0,a} - \ii\cos(2\theta)\ket{1,b}.
\end{equation}
Thus, $X_{\theta}$ has $e^{-\ii\pi/2}(\sin(2\theta),\cos(2\theta))$-revival between $(0,a)$ and $(1,b)$
at time $\pi/2$.
\end{proof}
\end{theorem}


\section{Basic Properties} \label{section:properties}

In this section, we describe several basic properties of graphs which have fractional revival.
First, we show a weak symmetry property for fractional revival.

\begin{proposition}\label{prop:sym}
If $(\alpha, \beta)$-revival occurs from $a$ to $b$ in a graph $X$ then
$(-\frac{\ov{\alpha}\beta}{\ov{\beta}}, \beta)$ revival occurs from $b$ to $a$ at the same time.
\begin{proof}
Suppose $U(\tau)\ee_a=\alpha \ee_a + \beta \ee_b$.   Since $U(\tau)$ is symmetric and unitary, we get
$U(\tau)_{a,b} = \beta$ and
\begin{equation}
( U(\tau)\ee_a)^{*} U(\tau) \ee_b = 0.
\end{equation}
The left-hand side can be written as
\begin{equation}
(\alpha \ee_a+\beta \ee_b)^{*} U(\tau)\ee_b = \ov{\alpha}\beta + \ov{\beta}  U(\tau)_{b,b},
\end{equation}
which gives 
\begin{equation}
U(\tau)_{b,b} = -\frac{\ov{\alpha}\beta}{\ov{\beta}}.
\end{equation}
We conclude from $|U(\tau)_{a,b}|^2+|U(\tau)_{b,b}|^2 = |\beta|^2+|\alpha|^2 =1$ that 
$(-\frac{\ov{\alpha}\beta}{\ov{\beta}}, \beta)$ revival occurs from $b$ to $a$ at time $\tau$.
\end{proof}
\end{proposition}

The above proposition allows us to say fractional revival {\em between vertices} $a$ {\em and} $b$, 
in place of fractional revival {\em from} $a$ {\em to} $b$.

Next, we describe a necessary spectral property of graphs with fractional revival.
Let $X$ be a graph with adjacency matrix $A(X)$ whose spectral decomposition is
\begin{equation}
A(X) = \sum_{r = 0}^d \theta_r E_r
\end{equation}
where $\theta_{r}$ are the eigenvalues and $E_{r}$ are the orthogonal projectors onto the
eigenspace corresponding to $\theta_{r}$.
Throughout this article, unless otherwise stated, 
we assume that $X$ is connected and $\theta_0$ is the largest eigenvalue of $A(X)$.

We say that two vertices $a$ and $b$ of $X$ are {\em parallel} if
$E_r \ee_a$ is parallel to $E_r \ee_b$ as vectors, for any $r$. 
The two vertices are called {\em cospectral} if $(E_{r})_{a,a} = (E_{r})_{b,b}$ for all $r$.
Finally, we say $a$ and $b$ are {\em strongly cospectral} if they are parallel and cospectral,
or that $E_{r}\ee_{a} = \pm E_{r}\ee_{b}$, for each $r$.
The {\em eigenvalue support} of $a$ is the set $\Phi_a = \{\theta_r : E_r\ee_a\neq 0\}$.
Strongly cospectral vertices have the same eigenvalue support.

\begin{proposition}
If there is $(\alpha,\beta)$-revival between $a$ and $b$ in a graph $X$, then
these vertices are parallel.
\begin{proof}
The equation $U(\tau)\ee_a=\alpha \ee_a + \beta \ee_b$ implies
\begin{equation}
e^{-\ii \tau \theta_r} E_r \ee_a = \alpha E_r \ee_a + \beta E_r \ee_b, \qquad \text{for all $r$.}
\end{equation}
This shows $a$ and $b$ are parallel since $\beta \neq 0$.
\end{proof}
\end{proposition}

Note that the results in this and the next sections apply to any  symmetric matrix $M$ with non-negative entries, in particular, to the adjacency matrix of a weighted undirected graph.

\begin{example}(Weighted Path on three vertices)\\
For $\omega \in \RR$,  let $P_3(\omega)$ denote the weighted path below.
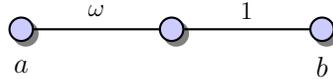
\begin{center}
\begin{tikzpicture}[
    scale=0.5,
    stone/.style={},
    black-stone/.style={black!80},
    black-highlight/.style={outer color=black!80, inner color=black!30},
    black-number/.style={white},
    white-stone/.style={white!70!black},
    white-highlight/.style={outer color=white!70!black, inner color=white},
    white-number/.style={black}]

\path (-2,+0.5) node (w1) [scale=0.8] {$\omega$};
\path (+2,+0.5) node (w2) [scale=0.8] {$1$};
\path (-4,-1.0) node (w1) [scale=0.9] {$a$};
\path (+4,-1.0) node (w1) [scale=0.9] {$b$};

\tikzstyle{every node}=[draw, thick, shape=circle, circular drop shadow, fill={blue!20}];
\path (-4,0) node (a) [scale=0.8] {};
\path (0,0) node (b) [scale=0.8] {};
\path (+4,0) node (c) [scale=0.8] {};
\draw[thick]
    (a) -- (b) -- (c);
\end{tikzpicture}
\end{center}
\ignore{
\begin{tikzpicture}[scale=1]
\path (0,0) coordinate (1);
\fill (1) circle (2pt);
\draw (0,0) node[anchor=north]{{\small $a$}};
\path (1,0) coordinate (2);
\fill (2) circle (2pt);
\path (2,0) coordinate (3);
\fill (3) circle (2pt);
\draw (2,0) node[anchor=north]{{\small $b$}};

\draw (1) -- (2);
\draw (2) -- (3);

\draw (0.5,0) node[anchor=south]{{\small $\omega$}};
\draw (1.5,0) node[anchor=south]{{\small $1$}};
\end{tikzpicture}
}
Then its adjacency matrix has three simple eigenvalues $0$ and $\pm\sqrt{\omega^2+1}$ with
\begin{eqnarray}
E_0 & = & \frac{1}{1+\omega^2} \begin{pmatrix} 1&0&-\omega\\0&0&0\\-\omega&0&\omega^2\end{pmatrix} \\
E_{\pm \sqrt{\omega^2+1}} & = & \frac{1}{2(\omega^2+1)}
\begin{pmatrix}
\omega^2 & \pm \omega\sqrt{\omega^2+1} & \omega\\
\pm \omega\sqrt{\omega^2+1} & \omega^2+1 & \pm\sqrt{\omega^2+1}\\
\omega& \pm\sqrt{\omega^2+1} &1
\end{pmatrix}.
 \end{eqnarray}
The transition matrix of the quantum walk at time $\tau=\frac{\pi}{\sqrt{\omega^2+1}}$ is
\begin{equation}
U(\tau) = 
\frac{1}{\omega^2+1} \begin{pmatrix} 1-\omega^2 & 0 & -2\omega\\0&-1-\omega^2&0\\-2\omega&0&\omega^2-1\end{pmatrix}.
\end{equation}
When $\omega=1$, the vertices $a$ and $b$ are strongly cospectral and there is  perfect state transfer between them at time $\frac{\pi}{\sqrt{2}}$.
For $\omega\neq 1$,  the vertices $a$ and $b$ are parallel but not cospectral.  The weighted path $P_3(\omega)$ has $(\frac{1-\omega^2}{1+\omega^2}, \frac{-2\omega}{1+\omega^2})$-revival between the
antipodal vertices.  
Note that $P_3(\sqrt{2}-1)$ has balanced fractional revival between $a$ and $b$.
\end{example}


\section{Cospectral Vertices} \label{section:cospectral}

We can say more when fractional revival occurs between cospectral vertices.
\begin{proposition}\label{prop:strongly_cospectral}
Suppose $(\alpha, \beta)$-revival occurs from $a$ to $b$ in a graph $X$ at time $\tau$.
Then $a$ and $b$ are (strongly) cospectral if and only if  there exist $\gamma, \zeta \in \RR$ such that
\begin{equation}
\alpha = e^{\ii\zeta}\cos \gamma 
\quad \text{and} \quad 
\beta = \ii e^{\ii \zeta}\sin \gamma.
\end{equation}
\begin{proof}
If $a$ and $b$ are cospectral then $U(\tau)_{a,a} = U(\tau)_{b,b}$.   From Proposition~\ref{prop:sym},  $\alpha \ov{\beta} \in \ii\RR$.
It follows from $|\alpha|^2+|\beta|^2=1$ that 
$\alpha = e^{\ii\zeta}\cos \gamma$ and 
$\beta =\ii e^{\ii \zeta}\sin \gamma$ for some $\gamma, \zeta \in \RR$.

Conversely, we have 
$U(\tau)\ee_a = e^{\ii \zeta}(\cos \gamma \ee_a + \ii \sin \gamma \ee_b)$ and 
$U(\tau)\ee_b = e^{\ii \zeta}(\cos \gamma \ee_b + \ii \sin \gamma \ee_a)$
which are equivalent to
\begin{align}
e^{-\ii \tau \theta_r} (E_r \ee_a - E_r \ee_b)  & =  e^{\ii \zeta}(\cos \gamma  - \ii \sin \gamma ) (E_r \ee_a - E_r \ee_b) \qquad \text{and} \\
e^{-\ii \tau \theta_r} (E_r \ee_a + E_r \ee_b)  & =  e^{\ii \zeta}(\cos \gamma  + \ii \sin \gamma ) (E_r \ee_a + E_r \ee_b)  
\end{align} 
for all $\theta_r\in \Phi_a$.
If there exists $r$ such that $E_r\ee_a \neq \pm E_r\ee_b$ then $\sin \gamma =0$.    It follows from $\beta\neq 0$ that $a$ and $b$ are strongly cospectral.
\end{proof}
\end{proposition}

Given strongly cospectral vertices $a$ and $b$ in a graph $X$,  define sets
\begin{equation}
\Phi_{a,b}^{+} = \{ \theta_r \ : \ E_r \ee_a =  E_r \ee_b\}
\quad\text{and}\quad
\Phi_{a,b}^{-} = \{ \theta_r \ : \ E_r \ee_a =  -E_r \ee_b\}.
\end{equation}
Note that $\Phi_a=\Phi_{a,b}^+ \cup \Phi_{a,b}^- = \Phi_b$.

\begin{theorem}\label{thm:equiv_cond}
Let $a$ and $b$ be strongly cospectral vertices in a graph $X$.
Then $(e^{\ii\zeta}\cos \gamma, \ii e^{\ii \zeta}\sin \gamma)$-revival occurs between $a$ and $b$ at time $\tau$ if and only if for all $\theta_r\in \Phi_a$
\begin{equation}\label{eqn:equiv_cond}
\tau(\theta_0 - \theta_r) =
\begin{cases}
0 \mod 2\pi & \text{if $\theta_r \in \Phi_{a,b}^+$,}\\
-2\gamma \mod 2\pi & \text{if $\theta_r \in \Phi_{a,b}^{-}$,}
\end{cases}
\end{equation}
and $\zeta = -\tau\theta_0-\gamma$.
\begin{proof}
We have  $U(\tau)\ee_a = e^{\ii \zeta}(\cos \gamma \ee_a + \ii \sin \gamma \ee_b)$ if and only if
\begin{equation} 
e^{-\ii\tau \theta_r} E_r \ee_a =
\begin{cases}
e^{\ii \zeta}(\cos \gamma  + \ii \sin \gamma) E_r \ee_a & \text{if $\theta_r \in \Phi_{a,b}^+$}\\
e^{\ii \zeta}(\cos \gamma  - \ii \sin \gamma) E_r \ee_a & \text{if $\theta_r \in \Phi_{a,b}^-$}
\end{cases}
\qquad \text{for all $\theta_r \in \Phi_a$.}
\end{equation}
Since, by Perron-Frobenius theorem, the largest eigenvalue $\theta_0$ belongs to $\Phi_{a,b}^+$, this is equivalent to
\begin{equation}
e^{\ii\tau (\theta_0-\theta_r)} =
\begin{cases}
1 & \text{if $\theta_r \in \Phi_{a,b}^+$}\\
\frac{\cos \gamma -\ii\sin \gamma}{\cos \gamma +\ii \sin\gamma} & \text{if $\theta_r \in \Phi_{a,b}^-$}
\end{cases}
\qquad \text{for all $\theta_r \in \Phi_a$.}
\end{equation}
The theorem follows from the fact that
\begin{equation}
\frac{\cos \gamma -\ii\sin \gamma}{\cos \gamma +\ii \sin\gamma} = e^{- 2 \gamma \ii}.
\end{equation}
\end{proof}
\end{theorem}

\begin{corollary}\label{cor:ratio}
Suppose $X$ admits fractional revival between strongly cospectral vertices $a$ and $b$.
Then 
\begin{equation}
\frac{\theta_i-\theta_j}{\theta_r-\theta_s} \in \QQ,
\end{equation}
for all  $\theta_i, \theta_j, \theta_r, \theta_s \in \Phi_{a,b}^+$, or  for all  $\theta_i, \theta_j, \theta_r, \theta_s \in \Phi_{a,b}^-$, with $\theta_r\neq \theta_s$.
\end{corollary}
When perfect state transfer occurs from $a$ to $b$ then $\gamma$ in Equation~(\ref{eqn:equiv_cond}) is an odd integer multiple of $\frac{\pi}{2}$ and $\tau(\theta_0-\theta_r) \in \pi \ZZ$, for all $\theta_r\in \Phi_a(=\Phi_b)$.   In this case, the above ratio is rational, for all $\theta_i, \theta_j, \theta_r, \theta_s \in \Phi_a$ with $\theta_r\neq \theta_s$.

Godsil \cite[Theorem 6.1]{g12b} showed that periodicity at a vertex implies that the eigenvalues in its eigenvalue support are integers or quadratic integers of a specific form. The key lemma to his result was the so called ratio condition \cite[Theorem 3.1]{g12b}. We extract from the proof of \cite[Theorem 6.1]{g12b} a lemma that shall be quite useful to us.

\begin{lemma}\label{lem:RatioConditionQuadraticIntegers}
Let $\Phi$ be a set of real algebraic integers which is closed under taking algebraic conjugates, and such that, for all $\theta_i,\theta_j,\theta_r,\theta_s \in \Phi$, we have
\[
\frac{\theta_i - \theta_j}{\theta_r - \theta_s} \in \QQ,
\ \hspace{0.5in} \
\mbox{ with $\theta_{r} \neq \theta_{s} \in \Phi$.}
\]
Then the elements of $\Phi$ are either integers or quadratic integers, and, moreover, if $|\Phi| = n$, then there are integers $a$, $\Delta$ (square-free), and $\{b_r\}_{r = 1}^n$ such that $\theta \in \Phi$ implies that, for some $r$,
\[\theta = \frac{a + b_r \sqrt{\Delta}}{2}.\] \qed
\end{lemma}

The proof works exactly as the proof of \cite[Theorem 6.1]{g12b} for when $|\Phi| \geq 3$. If $|\Phi| \leq 2$, the fact the $\Phi$ is closed under taking algebraic conjugates leads to an immediate proof.

If two vertices $a$ and $b$ are strongly cospectral, they share the same eigenvalue support, which is partitioned as $\Phi_{a,b}^+$ and $\Phi_{a,b}^-$. Corollary \ref{cor:ratio} above establish the ratio condition for each of these parts. As a consequence, we can derive the result below.

\begin{theorem} \label{thm:eigenvalsFR}
Assume $X$ admits fractional revival between strongly cospectral vertices $a$ and $b$. Let $\theta_0,...,\theta_t$ be the eigenvalues in their support. Then these are either integers or quadratic integers. 
Moreover, there are integers $a^+$, $a^-$, $\Delta^+$ (square-free), $\Delta^-$ (square-free), 
and $\{b_r\}_{r=0}^t$, such that, for all $r = 0,..,t$,
\begin{enumerate}[(i)]
\item if $\theta_r \in \Phi_{a,b}^+$, then $\displaystyle \theta_r = \frac{a^+ + b_r \sqrt{\Delta^+}}{2}$, and
\item if $\theta_r \in \Phi_{a,b}^-$, then $\displaystyle  \theta_r = \frac{a^- + b_r \sqrt{\Delta^-}}{2}$.
\end{enumerate}
\begin{proof}
From Corollary \ref{cor:ratio} and Lemma \ref{lem:RatioConditionQuadraticIntegers}, it is enough to show that each of these sets is closed under taking conjugates. Let $\theta$ be an eigenvalue which is not an integer and $\vv$ any of its eigenvectors. Let $\theta'$ be an algebraic conjugate of $\theta$. There is a field automorphism of
$\QQ(\theta)$ that maps $\theta$ to $\theta'$. The same automorphism, let us call it $\Psi$ now, maps $\vv$ to a vector $\vv'$ which is an eigenvector for the eigenvalue $\theta'$. As $a$ and $b$ are strongly cospectral, if follows that
\[\vv_a = \pm \vv_b.\]
Thus, with the same sign, we have
\[\vv'_a = \pm \vv'_b,\]
as $\Psi(1) = 1$ and $\Psi(-1) = -1$. Hence $\theta \in \Phi_{a,b}^+ \implies \theta' \in \Phi_{a,b}^+$ and $\theta \in \Phi_{a,b}^- \implies \theta' \in \Phi_{a,b}^-$ as wished.
\end{proof}
\end{theorem}

Some consequences ensue.

\begin{corollary}
Assume $X$ admits fractional revival between strongly cospectral vertices $a$ and $b$ at time $\tau$. With the notation of Theorem \ref{thm:eigenvalsFR}, let
\[g^+ = \gcd\left(\left\{\frac{\theta_r - \theta_s}{\sqrt{\Delta^+}}\right\}_{\theta_r,\theta_s \in \Phi_{a,b}^+}\right) \quad \text{and} \quad g^- = \gcd\left(\left\{\frac{\theta_r - \theta_s}{\sqrt{\Delta^-}}\right\}_{\theta_r,\theta_s \in \Phi_{a,b}^-}\right).\]
Then $\tau$ is an integer multiple of $2\pi/g^+\sqrt{\Delta^+}$ and also an integer multiple of $2\pi/g^-\sqrt{\Delta^-}$.

As a consequence, it follows that $\Delta^+ = \Delta^-$. 
\end{corollary}

\begin{proof}
This is a consequence of the fact that the differences of the eigenvalues are integers, and that Theorem \ref{thm:equiv_cond} implies that $\tau (\theta_r - \theta_s) \equiv 0 \pmod{2\pi}$ whenever $\theta_r,\theta_s$ are both in $\Phi_{a,b}^+$ or both in $\Phi_{a,b}^-$.
\end{proof}

The corollary below follows from an argument similar to the one in \cite[Corollary 6.2]{g12b}, taking into account Theorem \ref{thm:eigenvalsFR}.

\begin{corollary}
There are only finitely many connected graphs with maximum valency at most $k$ where fractional revival between cospectral vertices occur.
\end{corollary}

We see below that fractional revival between cospectral vertices implies either perfect state transfer, 
periodicity, or pretty good state transfer in the graph.

\begin{corollary}\label{cor:gammaQ}
Suppose $(e^{\ii\zeta}\cos \gamma, \ii e^{\ii \zeta}\sin \gamma)$-revival occurs between strongly cospectral vertices $a$ and $b$ in a graph $X$ at time $\tau$.
If $\gamma = \frac{p}{q} \pi$ for some coprime integers $p$ and $q$ then $X$ is periodic at both $a$ and $b$ at time $q\tau$.
If $q$ is even then $X$ admits perfect state transfer between $a$ and $b$ at time $\frac{q}{2}\tau$.
\begin{proof}
For any integer $k$, Equation(\ref{eqn:equiv_cond})  yields
\begin{equation}
k \tau(\theta_0 - \theta_r) =
\begin{cases}
0 \mod 2\pi & \text{if $\theta_r \in \Phi_{a,b}^+$}\\
-2\gamma k \mod 2\pi & \text{if $\theta_r \in \Phi_{a,b}^{-}$}
\end{cases}
\qquad \text{for all $\theta_{r} \in \Phi_a$.}
\end{equation}
When $\gamma=\frac{p}{q} \pi$,  Theorem~\ref{thm:equiv_cond} implies $U(q\tau) \ee_a = \pm e^{\ii q\zeta} \ee_a$.
If $\frac{q}{2} \in \ZZ$, then $p$ is odd and $U(\frac{q}{2}\tau) \ee_a = \pm e^{\ii \frac{q}{2} \zeta} \ii \ee_b$.
\end{proof}
\end{corollary}

\begin{corollary}
If balanced fractional revival occurs between strongly cospectral vertices $a$ and $b$ at time $\tau$ in a graph $X$,
then $X$ admits perfect state transfer between $a$ and $b$ at time $2\tau$, and $X$ is periodic at $a$ and $b$ 
at time $4\tau$.
\end{corollary}

When $\gamma$ is not a rational multiple of $\pi$,  we get pretty good state transfer between $a$ and $b$ as a consolation.
We say {\em pretty good (or almost perfect) state transfer} (see \cite{g12,vz12})
occurs between $a$ and $b$ if for all $\epsilon >0$, there exists a time $\tau'$  such that
\begin{equation}
|U(\tau') \ee_a - \nu \ee_b | < \epsilon,
\end{equation}
for some unimodular complex number $\nu$.
We need to use Kronecker's approximation theorem to prove this consolation.
\begin{theorem}[\cite{hw00}, see Theorem 442] \label{thm: kronecker} 
Let $1,\lambda_{1},\ldots,\lambda_{m}$ be linearly independent over $\Q$.
Let $\alpha_{1},\ldots,\alpha_{m}$ be arbitrary real numbers,
and let $N,\epsilon$ be positive real numbers.
Then there are integers $\ell > N$ and $q_1,\ldots, q_m$
so that
\begin{equation}\label{eq: kroneckers-theorem}
|\ell\lambda_{k} - q_{k} - \alpha_{k}| < \epsilon,
\end{equation}
for each $k=1,\ldots,m$.
\end{theorem}

\begin{corollary}\label{cor:pgst}
Suppose $(e^{\ii\zeta}\cos \gamma, \ii e^{\ii \zeta}\sin \gamma)$-revival occurs between strongly cospectral vertices $a$ and $b$ in a graph $X$ at time $\tau$.
If $\gamma = \lambda \pi$ for some irrational number $\lambda$ then $X$ admits pretty good state transfer between $a$ and $b$.
\begin{proof}
Applying Kronecker's approximation theorem to $\lambda_1= \frac{-2\gamma}{\pi}$ and $\alpha_1= \frac{1}{2}$,  there exist integers $\ell$ and $q_1$
satisfying
\begin{equation}
\frac{-2\gamma \ell}{\pi} \approx \frac{1}{2} + q_1.
\end{equation}
Equivalently,
$-4\gamma \ell \approx \pi+2q_1\pi$
and
\begin{equation}
2 \ell \tau(\theta_0 - \theta_r) \approx
\begin{cases}
0 \mod 2\pi & \text{if $\theta_r \in \Phi_{a,b}^+$}\\
\pi \mod 2\pi & \text{if $\theta_r \in \Phi_{a,b}^{-}$}
\end{cases}
\qquad \text{for all $r$.}
\end{equation}
It follow from Theorem~\ref{thm:equiv_cond} that $U(2\ell \tau) \ee_a \approx -\ii e^{\ii \zeta} \ee_b$.
\end{proof}
\end{corollary}

\begin{theorem}
If fractional revival occurs between two strongly cospectral vertices $a$ and $b$ in $X$,
then $X$ has perfect state transfer or pretty good state transfer from $a$ to $b$, or $X$ 
is periodic at $a$ and $b$ at the same time.
\end{theorem}

\ignore{
{\color{orange} Question for $H(n,2)$:  All vertices in any graph $X$ in $H(n,2)$ are strongly copsectral.     
If FR occurs in $X$ then it has to be $(e^{\ii\zeta}\cos \gamma, \ii e^{\ii \zeta}\sin \gamma)$-revival 
between antipodal vertices.   Since $X$ is integral, it is periodic and PGST implies PST.
Does it imply $\gamma \in \pi\QQ$?  A positive answer may be useful in our FR paper on schemes.}
}


\section{Equitable Partition} \label{section:equitable}

A partition $\rho = \bigcup_{i} C_{i}$ of the vertex set of a graph $X$ is {\sl equitable} if all vertices in the $i$-th cell $C_i$ of $\rho$ have the same number of neighbours, $d_{ij}$, in the $j$-th cell $C_j$ of $\rho$, for all $i, j$.
The quotient graph $X/\rho$ has its vertex set $\rho$ and the edge joining $C_i$ to $C_j$ has weight $\sqrt{d_{ij}d_{ji}}$.   

\begin{theorem}[Theorem 2 of \cite{bfflott}]
Suppose a graph $X$ has an equitable partition $\rho$ containing singleton cells $\{a\}$ and $\{b\}$.
Then 
\begin{equation}
(U_X(t))_{a,b} = U_{X/\rho}(t)_{\{a\},\{b\}}, \qquad \text{for all $t$}.
\end{equation}
\end{theorem}

\begin{corollary}\label{cor:quotient}
Suppose a graph $X$ has an equitable partition $\rho$ containing singleton cells $\{a\}$ and $\{b\}$.
There is fractional revival between $a$ and $b$ in $X$ if and only if there is fractional revival between $\{a\}$ and $\{b\}$ in the quotient $X/\rho$.
\end{corollary}

\begin{example}(Double Cones)\\
Let $Y$ be a $k$-regular graph on $n$ vertices.   Let $X$ be the join of $\ov{K_2}$ and $Y$, with the two vertices in $\ov{K_2}$ labeled as $a$ and $b$.
Then $\rho = \{ \{a\}, V(Y), \{b\}\}$ is an equitable partition of $X$ and the quotient graph $X/\rho$ has adjacency matrix
\begin{equation}
A(X/\rho) = \begin{pmatrix} 0 &\sqrt{n} & 0\\ \sqrt{n}& k &\sqrt{n}\\0&\sqrt{n}&0\end{pmatrix}.
\end{equation}
The eigenvalues of $A(X/\rho)$ are 
 $\theta_{\pm} = \frac{1}{2}(k\pm \sqrt{k^2+8n}) \in \Phi_{a,b}^+$ and $0\in \Phi_{a,b}^-$.
 Let $\tau = \frac{2\pi}{\sqrt{k^2+8n}}$ then
 \begin{equation}
\tau( \theta_+- \theta_-) =  2\pi \quad \text{and} \quad
\tau( \theta_+-0 ) = \frac{k+\sqrt{k^2+8n}}{\sqrt{k^2+8n}}\pi.
\end{equation}
It follows from Theorem \ref{thm:equiv_cond} that the quotient graph $X/\rho$ has $(e^{\ii \zeta} \cos \gamma, \ii e^{\ii \zeta} \sin \gamma)$-revival between cells $\{a\}$ and $\{b\}$ at time $\tau$ where 
\begin{equation}
\gamma =- \frac{k+\sqrt{k^2+8n}}{2\sqrt{k^2+8n}} \pi
\qquad \text{and} \qquad \zeta = \gamma.
\end{equation}
By Corollary \ref{cor:quotient}, the double cone $X$ admits fractional revival between $a$ and $b$.

In particular, the cocktail party graph $\ov{nK_2}$ admits fractional revival but not perfect state transfer between antipodal vertices, for all odd $n\geq 3$.   For even $n\geq 4$,  fractional revival occurs at time $\frac{\pi}{n}$ and perfect state transfer occurs at time $\frac{\pi}{2}$ in $\ov{nK_2}$ .
See Figure \ref{fig:cocktail_party}.
\end{example}

\begin{figure}[H]
\begin{center}
\begin{tikzpicture}[
    scale=0.5,
    stone/.style={},
    black-stone/.style={black!80},
    black-highlight/.style={outer color=black!80, inner color=black!30},
    black-number/.style={white},
    white-stone/.style={white!70!black},
    white-highlight/.style={outer color=white!70!black, inner color=white},
    white-number/.style={black}]


\tikzstyle{every node}=[draw, thick, shape=circle, circular drop shadow, fill={white}];
\path (-2,+1) node (a1) [scale=0.8] {};
\path (+2,-1) node (a2) [scale=0.8] {};
\tikzstyle{every node}=[draw, thick, shape=circle, circular drop shadow, fill={blue!20}];
\path (0,+2.5) node (b1) [scale=0.8] {};
\path (+2,+1) node (c1) [scale=0.8] {};
\path (0,-2.5) node (b2) [scale=0.8] {};
\path (-2,-1) node (c2) [scale=0.8] {};
\draw[thick]
    (a1) -- (b1) -- (c1) -- (a2) -- (b2) -- (c2) -- (a1)
    (a1) -- (c1) -- (b2) -- (a1)
    (b1) -- (a2) -- (c2) -- (b1);
\end{tikzpicture}
\quad \quad \quad
\begin{tikzpicture}[
    scale=0.5,
    stone/.style={},
    black-stone/.style={black!80},
    black-highlight/.style={outer color=black!80, inner color=black!30},
    black-number/.style={white},
    white-stone/.style={white!70!black},
    white-highlight/.style={outer color=white!70!black, inner color=white},
    white-number/.style={black}]


\tikzstyle{every node}=[draw, thick, shape=circle, circular drop shadow, fill={white}];
\path (-2.5,+1) node (a1) [scale=0.8] {};
\path (+2.5,-1) node (a2) [scale=0.8] {};
\tikzstyle{every node}=[draw, thick, shape=circle, circular drop shadow, fill={blue!20}];
\path (-1,+2.5) node (b1) [scale=0.8] {};
\path (+1,+2.5) node (c1) [scale=0.8] {};
\path (+2.5,+1) node (d1) [scale=0.8] {};
\path (+1,-2.5) node (b2) [scale=0.8] {};
\path (-1,-2.5) node (c2) [scale=0.8] {};
\path (-2.5,-1) node (d2) [scale=0.8] {};
\draw[thick]
    (a1) -- (b1) -- (c1) -- (d1) -- (a2) -- (b2) -- (c2) -- (d2) -- (a1)
    (a1) -- (c1) -- (a2) -- (c2) -- (a1)
    (b1) -- (d1) -- (b2) -- (d2) -- (b1)
    (a1) -- (d1) -- (c2) -- (b1) -- (a2) -- (d2) -- (c1) -- (b2) -- (a1);
\end{tikzpicture}
\caption{The double cones $\ov{nK_{2}}$ have fractional revival between antipodal vertices marked white.
Left: $\ov{3K_{2}}$ has fractional revival but no perfect state transfer. 
Right: $\ov{4K_{2}}$ has both fractional revival and perfect state transfer.}
\label{fig:cocktail_party}
\end{center}
\end{figure}
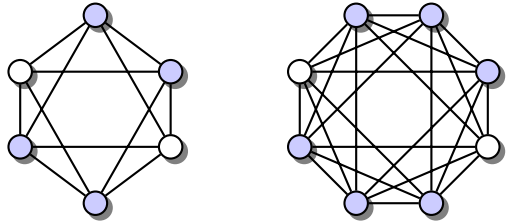

We use $Aut_X(a)$ to denote the group of automorphisms of $X$ that fix $a$. 
Let $\Delta_a$ be the coarsest equitable refinement of the partition $\{\{a\}, V(X)\backslash \{a\}\}$.
Then the orbits of $Aut_X(a)$ is an equitable refinement of $\Delta_a$.
Just as when perfect state transfer occurs between $a$ and $b$, see \cite{g12b}, it follows from the following proposition that if fractional revival occurs between
$a$ and $b$ then $\Delta_a=\Delta_b$.
\begin{proposition}\label{prop:aut}
Suppose fractional revival occurs between $a$ and $b$ in a graph $X$,   then $Aut_X(a) = Aut_X(b)$.
\begin{proof}
Suppose $(\alpha,\beta)$-revival occurs from $a$ to $b$ at time $\tau$.
Let $P$ be the permutation matrix of  an automorphism in $Aut_X(a)$ that maps $b$ to $c$.   We have
\begin{equation}
\ee_c^* U(\tau) \ee_a = \ee_c^* U(\tau) P \ee_a = \ee_b^* U(\tau) \ee_a =\beta.
\end{equation}
We conclude that $\ee_c^* \ee_b \neq 0$ and $b=c$.
\end{proof}
\end{proposition}

Consequently, if  fractional revival occurs in a cycle then the cycle has even length and fractional revival occurs between antipodal vertices.


\section{Bipartite Graphs and Even Cycles} \label{section:bipartite}

Let $X$ be a connected bipartite graph with adjacency matrix
\begin{equation}
A = 
\begin{pmatrix}
0 & B\\ B^T & 0
\end{pmatrix}
\end{equation}
Then the transition matrix is
\begin{equation}
U(t) = \sum_{k\geq 0} \frac{(-1)^k t^{2k}}{(2k)!} \begin{pmatrix} (BB^T)^k &0\\0 & (B^TB)^k\end{pmatrix} -
\ii \sum_{k\geq 0} \frac{(-1)^k t^{2k+1}}{(2k+1)!} \begin{pmatrix} 0&(BB^T)^kB\\ (B^TB)^kB^T&0\end{pmatrix}.
\end{equation}

\begin{figure}[t]
\begin{center}
\begin{tikzpicture}[
    scale=0.5,
    stone/.style={},
    black-stone/.style={black!80},
    black-highlight/.style={outer color=black!80, inner color=black!30},
    black-number/.style={white},
    white-stone/.style={white!70!black},
    white-highlight/.style={outer color=white!70!black, inner color=white},
    white-number/.style={black}]


\tikzstyle{every node}=[draw, thick, shape=circle, circular drop shadow, fill={white}];
\path (0,+2) node (a) [scale=0.8] {};
\path (0,-2) node (c) [scale=0.8] {};
\tikzstyle{every node}=[draw, thick, shape=circle, circular drop shadow, fill={blue!20}];
\path (+2,0) node (b) [scale=0.8] {};
\path (-2,0) node (d) [scale=0.8] {};
\draw[thick]
    (a) -- (b) -- (c) -- (d) -- (a);
\end{tikzpicture}
\quad \quad \quad
\begin{tikzpicture}[
    scale=0.5,
    stone/.style={},
    black-stone/.style={black!80},
    black-highlight/.style={outer color=black!80, inner color=black!30},
    black-number/.style={white},
    white-stone/.style={white!70!black},
    white-highlight/.style={outer color=white!70!black, inner color=white},
    white-number/.style={black}]


\tikzstyle{every node}=[draw, thick, shape=circle, circular drop shadow, fill={white}];
\path (-2,+1) node (a) [scale=0.8] {};
\path (+2,-1) node (d) [scale=0.8] {};
\tikzstyle{every node}=[draw, thick, shape=circle, circular drop shadow, fill={blue!20}];
\path (0,+2) node (b) [scale=0.8] {};
\path (+2,+1) node (c) [scale=0.8] {};
\path (0,-2) node (e) [scale=0.8] {};
\path (-2,-1) node (f) [scale=0.8] {};
\draw[thick]
    (a) -- (b) -- (c) -- (d) -- (e) -- (f) -- (a);
\end{tikzpicture}
\caption{The only cycles with fractional revival: $C_{4}$ and $C_{6}$.
Fractional revival occurs between antipodal vertices marked white.}
\label{fig:cycles}
\end{center}
\end{figure}
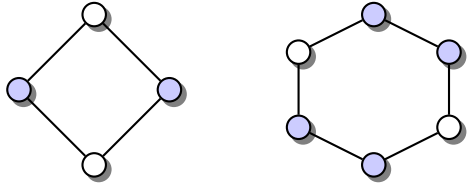

\begin{theorem}\label{thm:bipartite}
Suppose $(\alpha, \beta)$-revival occurs between $a$ and $b$ in a bipartite graph $X$ at time $\tau$.
If $a$ and $b$ belong to different bipartition of $X$ then $a$ and $b$ are strongly cospectral.   If $a$ and $b$ belong to the same bipartition then $X$ is periodic at both vertices at time $2\tau$. 
\begin{proof}
Observe that  if $a$ and $b$ belong to different bipartitions of $X$ then $\alpha=\cos \gamma$ and $\beta=\ii \sin \gamma$ for some $\gamma \in \RR$.   We conclude from Proposition~\ref{prop:strongly_cospectral} that $a$ and $b$ are strongly cospectral.

If $a$ and $b$ belong to the same bipartition of $X$ then either $\alpha, \beta \in \RR$ or $\alpha, \beta \in \ii \RR$.  
From Proposition~\ref{prop:sym}, we see that $U(\tau)_{b,b} = -\alpha$.
We reorder the vertices so that the first two rows and columns of $A(X)$ are indexed by $a$ and $b$.
Then 
\begin{equation}
U(\tau) = \begin{pmatrix} M_1 & 0^T\\0 & M_2\end{pmatrix}
\end{equation}
where
\begin{equation}
M_1 = \begin{pmatrix} \alpha & \beta \\ \beta & -\alpha \end{pmatrix}
\end{equation}
and $M_2$ is a symmetric unitary matrix.
Since $M_1^2 = \pm I$,  we conclude that $X$ is periodic at $a$ and $b$ at time $2\tau$.
\end{proof}
\end{theorem}

In the latter case, the following ratio condition \cite{g12b} holds for the eigenvalues in $\Phi_a$ or in $\Phi_b$.
\begin{theorem} \label{thm:ratio_periodic}
If $X$ is periodic at $a$, then
\begin{equation}
\frac{\theta_i-\theta_j}{\theta_r-\theta_s} \in \QQ
\end{equation}
for any $\theta_i, \theta_j, \theta_r, \theta_s \in \Phi_a$ with $\theta_r\neq \theta_s$.
\end{theorem}

\begin{example} (Cycles)\\
It is known that $C_4$ has perfect state transfer at time $\frac{\pi}{2}$.   
In $C_6$, antipodal vertices $a$ and $b$ are strongly cospectral, and
\begin{equation}
\Phi_{a,b}^+=\{2,-1\}
\quad \text{and} \quad \Phi_{a,b}^- = \{1,-2\}.
\end{equation}
By Theorem~\ref{thm:equiv_cond},  $(-\frac{1}{2}, \frac{\sqrt{3}}{2} \ii)$-revival occurs between $a$ and $b$ in $C_6$ at time $\frac{2\pi}{3}$.
\end{example}

We show that no other cycle admits fractional revival using the following number theory result \cite{berger}.
Given a non-zero rational number $\mu$, we use $N(\mu)$ to denote 
$q$ if $\mu=\frac{p}{q}$ where $p$ and $q$ are coprime and $q>0$.
\begin{theorem}\label{thm:cosine}
Suppose $\mu_1$ and $\mu_2$ are rational numbers such that  $\mu_1\pm \mu_2$ are not integers.
Then  the set $\{1, \cos(\mu_1 \pi), \cos(\mu_2 \pi)\}$ is linearly independent over $\QQ$
if and only if $N(\mu_1), N(\mu_2) \geq 4$ and $(N(\mu_1), N(\mu_2)) \neq (5,5)$.
\end{theorem}

\begin{theorem}
Fractional revival occurs in a cycle if and only if it has four or six vertices.
\begin{proof}
From Proposition \ref{prop:aut}, it is sufficient to consider only the antipodal vertices in even cycles.

Suppose $a$ and $b$ are antipodal vertices in an even cycle $C_n$.
Then 
\begin{equation}
\Phi_{a,b}^+ = \{ 2\cos (\frac{2r\pi}{n}) \ : \ \text{$r$ is even}\} 
\quad \text{and} \quad
\Phi_{a,b}^-= \{ 2\cos (\frac{2r\pi}{n}) \ : \ \text{$r$ is odd}\}.
\end{equation}

If $n = 2\mod 4$, then $a$ and $b$ are in different bipartitions of $C_n$.  
When $n \neq 6, 10$, we have $N(\frac{4}{n})=N(\frac{8}{n}) \geq 7$, and by Theorem \ref{thm:cosine}, the set $\{1, \cos(\frac{4}{n}\pi), \cos(\frac{8}{n} \pi)\}$ is linearly independent over $\QQ$.
Hence $\frac{2-2\cos(4\pi/n)}{2-2\cos(8\pi/n)} \not \in \QQ$.  By Corollary~\ref{cor:ratio}, there is no fractional revival between $a$ and $b$.

In $C_{10}$, the eigenvalues  $2$ and $ \frac{-1\pm\sqrt{5}}{2}$ belong to $\Phi_{a,b}^+$.   By Corollary~\ref{cor:ratio}, the occurrence of fractional revival between $a$ and $b$ implies
\begin{equation*}
\frac{2- (\frac{-1+\sqrt{5}}{2})}{(\frac{-1+\sqrt{5}}{2}) - (\frac{-1-\sqrt{5}}{2})} \in \QQ,
\end{equation*}
or  $\sqrt{5} \in \QQ$.

If $n = 0 \mod 4$, then $a$ and $b$ are in same bipartition of $C_n$.   For $n \geq 16$,  Theorem~\ref{thm:cosine} implies the set $\{2, 2\cos(\frac{2}{n}\pi), 2\cos(\frac{4}{n} \pi)\}$ is linearly independent over 
$\QQ$.  Therefore the ratio condition in Theorem~\ref{thm:ratio_periodic} cannot hold and $C_n$ is not periodic.   Hence $C_n$ has no fractional revival for $n \geq 16$.

When $n=8$ or $12$,  the ratio $\frac{2-2\cos \frac{2\pi}{n}}{2-(-2)} \not \in \QQ$.   We conclude that  $C_8$ and $C_{12}$ are not periodic and they cannot have fractional revival.
\end{proof}
\end{theorem}


\section{Paths} \label{section:path}

In this section, we determine the paths that admit fractional revival.

Recall that $P_n$ denotes the path on vertices $\{1,2,\dots, n\}$ where $a$ is adjacent to $a+1$, 
for $a = 1,\ldots, n-1$.   
If $n$ is odd and $a$ is the middle vertex on $P_n$ then
$|Aut_{P_n}(a)|=2 \neq |Aut_{P_n}(j)|$, for $j\neq a$.   If follows from Proposition~\ref{prop:aut} that $P_n$ has no fractional revival involving $a$.

The eigenvalues of $P_n$ are 
\begin{equation}
\theta_r=2\cos(\frac{\pi r}{n+1}), \quad r=1,...,n,
\end{equation}
with
\begin{equation}\label{eqn:path_idem}
(E_r)_{j,a} = \frac{2}{n+1} \sin(\frac{jr\pi}{n+1})\sin(\frac{ar\pi}{n+1}),
\end{equation}
for $r,j,a=1,\ldots,n$.  See Section 1.4.4 of \cite{bh}.

\begin{figure}[t]
\begin{center}
\begin{tikzpicture}[
    scale=0.5,
    stone/.style={},
    black-stone/.style={black!80},
    black-highlight/.style={outer color=black!80, inner color=black!30},
    black-number/.style={white},
    white-stone/.style={white!70!black},
    white-highlight/.style={outer color=white!70!black, inner color=white},
    white-number/.style={black}]


\tikzstyle{every node}=[draw, thick, shape=circle, circular drop shadow, fill={white}];
\path (-1,0) node (a) [scale=0.8] {};
\path (+1,0) node (b) [scale=0.8] {};
\draw[thick]
    (a) -- (b);
\end{tikzpicture}
\quad \quad \quad
\begin{tikzpicture}[
    scale=0.5,
    stone/.style={},
    black-stone/.style={black!80},
    black-highlight/.style={outer color=black!80, inner color=black!30},
    black-number/.style={white},
    white-stone/.style={white!70!black},
    white-highlight/.style={outer color=white!70!black, inner color=white},
    white-number/.style={black}]


\tikzstyle{every node}=[draw, thick, shape=circle, circular drop shadow, fill={white}];
\path (-2,0) node (a) [scale=0.8] {};
\path (+2,0) node (c) [scale=0.8] {};
\tikzstyle{every node}=[draw, thick, shape=circle, circular drop shadow, fill={blue!20}];
\path (0,0) node (b) [scale=0.8] {};
\draw[thick]
    (a) -- (b) -- (c);
\end{tikzpicture}
\quad \quad \quad
\begin{tikzpicture}[
    scale=0.5,
    stone/.style={},
    black-stone/.style={black!80},
    black-highlight/.style={outer color=black!80, inner color=black!30},
    black-number/.style={white},
    white-stone/.style={white!70!black},
    white-highlight/.style={outer color=white!70!black, inner color=white},
    white-number/.style={black}]


\tikzstyle{every node}=[draw, thick, shape=circle, circular drop shadow, fill={white}];
\path (-3,0) node (a) [scale=0.8] {};
\path (+3,0) node (d) [scale=0.8] {};
\tikzstyle{every node}=[draw, thick, shape=circle, circular drop shadow, fill={blue!20}];
\path (-1,0) node (b) [scale=0.8] {};
\path (+1,0) node (c) [scale=0.8] {};
\draw[thick]
    (a) -- (b) -- (c) -- (d);
\end{tikzpicture}
\caption{The only paths with fractional revival: $P_{2}$, $P_{3}$ and $P_{4}$.
Fractional revival occurs between antipodal vertices marked white.}
\label{fig:paths}
\end{center}
\end{figure}
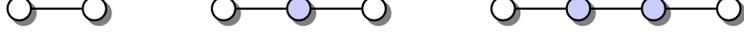

We observe that the $(n+1-a)$ is the only vertex strongly cospectral with the vertex $a$ in $P_n$,
and if $\theta_r \in \Phi_{a,n+1-a}^+$ then $r$ is odd.

It follows from Equation~(\ref{eqn:path_idem}) that
\begin{equation}
\Phi_a = \{ r \ :\ (n+1) \not | ar\}.
\end{equation}
Suppose $a$ is not the middle vertex of $P_n$ for some odd $n$, then
\begin{equation}
\theta_1, \theta_2, \theta_n \in \Phi_a.
\end{equation}
The ratio 
\begin{equation}
\frac{\theta_1-\theta_2}{\theta_1-\theta_n} 
= \frac{\theta_1-\theta_2}{2\theta_1} =\frac{1}{2} - \frac{\cos \frac{2\pi}{n+1}}{2\cos \frac{\pi}{n+1}}.
\end{equation}
When $n\geq 5$ then $\theta_2 \neq 0$ and 
the set
$\{1, \cos \frac{\pi}{n+1}, \cos \frac{2\pi}{n+1}\}$ is linearly independent over $\QQ$.   We conclude that the above ratio is irrational.
From Theorem~\ref{thm:ratio_periodic}, $a$ is not periodic in $P_n$.

Suppose fractional revival occurs between $a$ and $b$ in $P_n$, for $n \geq 5$.
By Theorem~\ref{thm:bipartite}, $a$ and $b$  belong to different bipartitions of $P_n$ and they are strongly cospectral.   This can happen only when $n$ is even and $b=n+1-a$.
Since $a$ and $b$ are not periodic vertices, Corollary~\ref{cor:gammaQ} implies that $(\cos \gamma, \ii \sin \gamma)$-revival occurs from $a$ to $b$ for some irrational number $\gamma$.  By Corollary~\ref{cor:pgst},  $P_n$ admits pretty good state transfer.
It is known that $P_n$ admits pretty good state transfer if and only if 
$n+1$ is a power of 2, $p$, or $2p$, for some odd prime $p$, see \cite{gkss}.
We conclude that $n+1=p$ for some odd prime.

For $n+1=p > 5$, $\theta_1, \theta_3, \theta_5\in \Phi_{a,n+1-a}^+$ and
the ratio
\begin{equation}\label{eqn:ratioP_n}
\frac{\theta_1-\theta_3}{\theta_3-\theta_5}
= \frac{\sin \frac{\pi}{n+1}\sin \frac{2\pi}{n+1}}{\sin \frac{\pi}{n+1}\sin\frac{4\pi}{n+1}}
=\frac{\cos (\frac{\pi}{2}-\frac{2\pi}{n+1})}{\cos (\frac{\pi}{2}-\frac{4\pi}{n+1})}
=\frac{\cos \frac{(n-3)\pi}{2(n+1)}}{\cos\frac{(n-7)\pi}{2(n+1)}}.
\end{equation}
As $n$ is even and $n+1$ is a prime, $N(\frac{n-3}{2(n+1)})=N(\frac{n-7}{2(n+1)}) = 2(n+1) > 5$ 
and $\frac{n-3}{2(n+1)} \pm \frac{n-7}{2(n+1)} \not \in \ZZ$.
By Theorem~\ref{thm:cosine},  the ratio in Equation~(\ref{eqn:ratioP_n}) is irrational.
It follows from Corollary~\ref{cor:ratio} that there is no fractional revival between $a$ and $(n+1-a)$.   Hence there is no fractional revival in $P_n$ for $n\geq 5$.

It is known that $P_2$ admits fractional revival at time $\tau$ when $\tau \not \in \pi\ZZ$, and $P_3$ admits perfect state transfer at time $\frac{\pi}{\sqrt{2}}$.

In $P_4$, for any vertex $a$
\begin{equation}
\Phi_{a,5-a}^+ = \{\frac{1+\sqrt{5}}{2}, \frac{1-\sqrt{5}}{2}\}
\quad \text{and}\quad
\Phi_{a,5-a}^- = \{\frac{-1+\sqrt{5}}{2}, \frac{-1-\sqrt{5}}{2}\}.
\end{equation}
It follows from Theorem~\ref{thm:equiv_cond} that $(-\cos \frac{\pi}{\sqrt{5}},\ii \sin \frac{\pi}{\sqrt{5}})$-revival occurs between vertices $a$ and $(5-a)$, for $a=1,2$ at time $\frac{2\pi}{\sqrt{5}}$.

\begin{theorem}
A path $P_n$ admits fractional revival if and only if $n \in \{2,3,4\}$.
\end{theorem}

\section{Conclusion}

The theme of this paper is fractional revival in simple unweighted graphs according to the adjacency matrix model. We have developed tools that allowed for the construction of new examples (Section \ref{section:construction}), for a better understanding of the underlying theoretical framework (Sections \ref{section:properties}, \ref{section:cospectral} and \ref{section:equitable}) and for the analysis of fractional revival in large families of graphs (Sections \ref{section:bipartite} and \ref{section:path}). This work offers a self-contained complete introduction to the topic, that unveils connections to known results in related fields as well as presents new advances. Nonetheless, there are still many questions that would deserve further exploration. For instance:
\begin{enumerate}
\item In Section \ref{section:cospectral}, a deep theoretical analysis of fractional revival was made for when the vertices involved are cospectral. We have observed however that, unlike perfect or pretty good state transfer, fractional revival may occur amongst non-cospectral vertices. An open line of investigation is to study deeper properties that such vertices must satisfy. For example: is there a limit to the degrees of the field extensions that contain the eigenvalues in their support?
\item Find examples of fractional revival between non-cospectral vertices (in simple unweighted graphs).
\item If pretty good state transfer occurs between $a$ and $b$, is it possible to characterize in simple terms the cases in which fractional revival also occurs between these vertices? For example, there are infinitely many paths admitting pretty good state transfer (see \cite{gkss}), but only three of them have of fractional revival.
\item How about pretty good fractional revival? In other words, what can be said for when the probability of observing $\ket 1$ is almost entirely concentrated in two vertices?
\item Fractional revival in the Laplacian matrix model has not been explored (to the best of our knowledge). For example, are there trees other than $P_2$ in which fractional revival occurs between cospectral vertices? This question about trees can also be asked for the adjacency matrix model, but we believe the latter question will be quite harder than the former.
\item Are there examples of graphs in which three (or more) vertices are involved in fractional revival, two at time? In other words, is fractional revival monogamous like its cousin perfect state transfer is (see \cite{k11basics})?
\item Study fractional revival in association schemes. We have a work in preparation that will cover this topic.
\end{enumerate}


\section*{Acknowledgments}

We thank Chris Godsil for helpful comments. 
G.C. thanks a travel grant awarded by the Department of Computer Science at UFMG. 
The research of L.V. is supported by a discovery grant from the National Science and 
Engineering Research Council (NSERC) of Canada.
C.T. would like to thank Centre de Recherches Math\'{e}matiques, Universit\'e de Montr\'eal
for its kind hospitality.



\end{document}